\documentclass[10pt]{IEEEtran}
\usepackage{amsmath, mathrsfs,amsfonts}
\usepackage{amssymb}
\usepackage{acronym}
\usepackage{graphicx, amsthm}

\def\bs{\boldsymbol}

\def \foral {\textrm{for all }}

\newtheorem{proposition}{Proposition}

\theoremstyle{definition}

\newtheorem{definition}{Definition}

               {\begin{list}{}{\leftmargin#1\rightmargin#2\topsep#3}\item{}}%
               {\end{list}}

\title{Optimal Reverse Carpooling Over Wireless Networks - A Distributed Optimization Approach}
\author{Ali ParandehGheibi$^\dag$\thanks{$\dag$ Department of Electrical Engineering and Computer Science, MIT.},  Asuman Ozdaglar$^\dag$, Michelle Effros$^\ddag$\thanks{$\ddag$ Department of Electrical Engineering, California Institute of Technology.}, Muriel M\'edard$^\dag$\\
\normalsize{parandeh@mit.edu, asuman@mit.edu, effros@caltech.edu, medard@mit.edu}}

\begin{document}

\maketitle

\begin{abstract}
We focus on a particular form of network coding, reverse carpooling, in a wireless
network where the potentially coded transmitted messages are to be decoded immediately upon
reception. The network is fixed and known, and the system performance is measured in terms of the number of wireless broadcasts required to meet multiple unicast demands. Motivated by the structure of the coding scheme, we formulate the problem as a linear program by introducing a flow
variable for each triple of connected nodes. This allows us to have a formulation polynomial in the number of nodes.  Using dual decomposition and projected subgradient method, we present a
decentralized algorithm to obtain optimal routing schemes in presence of coding opportunities. We show that the primal sub-problem can be expressed as a shortest path problem on an \emph{edge-graph}, and the proposed algorithm requires each node to exchange information only with its neighbors.

Keywords: Network coding, Distributed optimization.
\end{abstract}

\section{Introduction}\label{introduction_sec}
The network coding literature has provided a new set of tools for designing wireless networks. The early works in this literature were mainly focused on the multicast problem, where all sinks in the network require all of the information transmitted at the source node. Linear network coding is sufficient to achieve the capacity of a network for a multicast session \cite{algebraic}. Another interesting problem is the multiple unicast problem, in which each source delivers its information to the corresponding destination. In this scenario, the coding schemes could be more complicated. In fact, linear network coding is not sufficient to achieve the capacity of a network with multiple unicast sessions \cite{insufficiency}. Nevertheless, simple linear network coding techniques can boost the performance of the network significantly. An example of a simple coding opportunity across multiple sessions is \emph{reverse carpooling} (see, for example \cite{tiling}). Consider the three-node wireless network configuration depicted in Figure \ref{carpool_fig}, where nodes A and B can only communicate via the relay node (R). The side nodes attempt to exchange messages $P_A$ and $P_B$. Routing schemes require two transmissions at the relay node to exchange the messages. However, by allowing network coding we can save one transmission as follows. The relay node first receives the packets $P_A$ and $P_B$ from A and B, respectively. Then it forms the bitwise XOR of the packets and broadcasts the \emph{coded} packet. Consequently, node A (B) can decode packet $P_B$ ($P_A$) by XORing the received coded packet with the packet $P_A$ ($P_B$) that it transmitted. This technique can be used to decrease the number of wireless transmissions (energy) when two unicast flows share a common path in opposite directions. Reverse carpooling is not the only form of opportunistic coding strategies. Katti \emph{et. al.} \cite{cope} present COPE, a practical scheme that exploits the overheard packets to introduce new coding opportunities. Experimental results in \cite{cope} demonstrate that COPE significantly increases the network throughput.
Note that the broadcast nature of the wireless medium is the essential ingredient for such coding opportunities.

\begin{figure}
\centering
  \includegraphics[width=.5\textwidth]{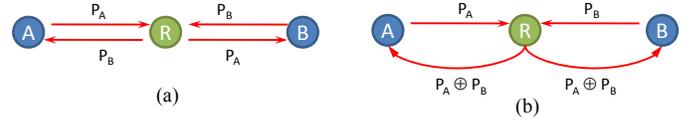}\\
  \caption{(a) Relay node requires two transmissions without network coding. (b) Relay node requires a single broadcast with network coding}\label{carpool_fig}
\end{figure}

In this work, we study the multiple unicast problem over a wireless network where only routing and  reverse carpooling are allowed. The goal is to minimize the total cost of transmissions (energy expenditure) to support the unicast sessions.

Given the decentralized nature of the information, we are interested in distributed algorithms that can implement the optimal policies using coordination only among neighboring nodes. Note that without coding, this problem is equivalent to finding the shortest path from each source to the corresponding destination, which can be solved efficiently in a distributed manner by storing a \emph{routing table} at each node and updating each table via local communication. Here, the routing decisions for each unicast session is independent of those for other sessions. In contrast, when reverse carpooling is allowed, the best path choice depends on the existence and rate of other unicast sessions. For example, consider the two unit-rate unicast sessions illustrated in Figure \ref{grid_unicast_fig}(a). In this example, deviation from the shortest path increases coding opportunities and hence decreases the total cost. Now, suppose the rate of session 2 is increased while the rate of session 1 is fixed. In this case, only session 1 deviates from its shortest path to reverse carpool with session 2. In this paper, using dual decomposition methods \cite{dist_book}, we design a pricing mechanism at each node that allows us to obtain optimal routing decisions. Moreover, we show that optimal prices (dual variables) and routing decisions (primal variables) can be computed efficiently by exchanging information solely with immediate neighbors.

\begin{figure}
\centering
  \includegraphics[width=.31\textwidth]{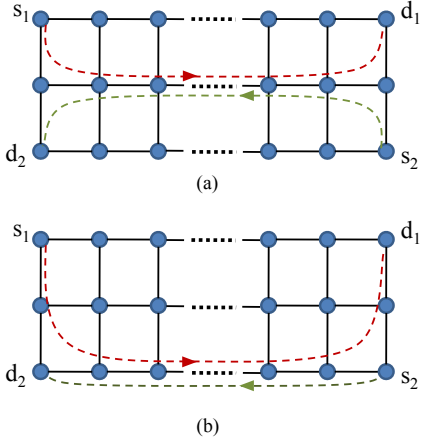}\\
  \caption{(a) Both unit-rate unicast sessions deviate from the shortest path (direct line) to maximize coding opportunities. (b) The optimal routing decisions change when rate of session 2 is increased.}\label{grid_unicast_fig}
\end{figure}

Related literature addresses the opportunistic coding problem from different points of view. Marden and Effros \cite{marden} use a game theoretic approach which has become popular in the recent literature for devising distributed control policies for networks. This entails viewing each source destination pair as a selfish agent and designing a utility function for each one. It is well-known that such a multi-agent model creates incentive issues and leads to loss in performance. The authors of \cite{marden} characterize suboptimality of the solutions reached by acting selfishly compared to the globally optimal solution. In contrast, we obtain the globally optimal solution using only limited coordination among the nodes. Also related to our work is Sengupta \emph{et al.} \cite{cope_optimization}, which take a centralized optimization approach and provide various linear optimization formulations for a network with restricted forms of network coding. The formulations presented in \cite{cope_optimization} are based on all paths between each source destination pair, which lead to exponential-sized problems. In another related work, Reddy \emph{et al.} \cite{srinivas} adopt a path-based formulation together with population game dynamics to achieve a close approximation of the optimal policies. Our problem formulation, in contrast, is polynomial in the number of nodes of the network, and hence we are able to obtain the globally optimal solution in polynomial time.
%

\section{System Model}\label{model_sec}
%

Consider a wireless network with $n$ nodes modeled as an undirected graph $G = (\mathcal N, \mathcal E)$, where $\mathcal N= \{1, \ldots, n\}$, and $\{i,j\} \in \mathcal E$ if and only if  $i$ and $j$ can directly communicate to each other. We refer to such nodes as \emph{neighbors}. We assume that the broadcast
constraint is present, i.e, each transmission by a particular node is received by all of its
neighbors. We also assume that  the interference at the receivers is avoided by a Medium Access Control (MAC) protocol such as CSMA.

We study a multiple unicast scenario where there is a set $T$ of unicast flows. Each unicast
flow $t \in T$ corresponds to a source-destination pair $(s_t, d_t)$, and a demand $R_t$ which is
the rate at which packets from source $s_t$ should be delivered to destination $d_t$ reliably. Throughout this work we assume the set of unicast connections are feasible, i.e., there exists at least one path connecting each source to its destination, and no capacity constraint exists.

For each node $i$, the cost of broadcasting a unit-rate flow (one packet per unit time) is denoted by $c_i$. This could reflect the energy expenditure or the level of contention at a particular node. Throughout this work, we assume that only routing and reverse carpooling coding operations (as described in Section \ref{introduction_sec}) are allowed. The goal is to minimize the cost required to support all unicast connections in steady state. Next, we present a linear programming formulation for this problem.

\section{Problem Formulation}\label{formulation_sec}
In this section, we discuss how to formulate the problem of coding-aware routing for multiple
unicast sessions as a linear optimization problem. In general, there are two approaches to formulate a multiple-unicast problem which is also referred to as \emph{multi-commodity flow problem} in the network optimization literature: The \emph{path-based} approach and the \emph{link-based} approach. The path-based approach considers all possible paths connecting each source-destination pair. This approach in general leads to a formulation of exponential size. Moreover, this approach is not suitable for distributing the decision making among individual wireless nodes. On the other hand, the link-based approach is more appealing from decentralized optimization point of view as well as the problem size. However, when reverse carpooling is allowed, there is no simple way of writing the transmission cost of a particular node in terms of the flow variables on its incoming and outgoing links. That is because in the case of reverse carpooling there is no one-to-one correspondence between packet transfers and physical transmissions.

In order to have a finer description of the flows, we define the decision variables to describe the flow of packets over two successive hops. Let $i$ be a node with neighbors $v$ and $w$. Denote by $x^t(v,i,w)$ the flow of packets of the session $t \in T$, transferred from $v$ to $w$ via node $i$. When reverse carpooling is allowed, node $i$ can transfer $\sum_{t \in T} x^t(v,i,w)$ packets from $v$ to $w$, and $\sum_{t \in T} x^t(w,i,v)$ packets from $w$ to $v$ with
$$y^i_{vw} = y^i_{wv} =  \max \Big\{ \sum_{t \in T}x^t(v,i,w) , \sum_{t \in T}x^t(w,i,v)\Big\} $$
transmissions. Since all reverse carpooling opportunities are along such triples of nodes, the total number of transmissions of a particular relay node $i$ of graph $G$ is
\begin{equation}\label{tx_count}
z_i = \sum_{v<w: (v,i,w) \in \Gamma_G} y^i_{vw},
\end{equation}
where $\Gamma_G = \Big\{(i,j,k): \{i,j\} \in \mathcal E, \{j,k\} \in \mathcal E, i \neq k \Big\}$.

Note that every node can act as a relay, source or destination node. In order to write the constraints in a compact way, we define an \emph{expanded graph} in which we introduce an artificial source and destination node for each unicast session.

\begin{definition}
Given the network connectivity graph $G = (\mathcal N, \mathcal E)$, and a set of unicast connections described as $\big\{(s_t, d_t, R_t): t \in T\big\}$, the \emph{expanded graph} is an undirected graph $G' = (\mathcal N', \mathcal E')$ where
$$\mathcal N' = \mathcal N \cup \bigg(\bigcup_{t \in T} \{s'_t, d'_t\}\bigg),$$
$$\mathcal E' = \mathcal E \cup \bigg(\bigcup_{t \in T} \Big\{ \{s'_t, s_t\}, \{d'_t, d_t\}\Big\}   \bigg),$$
and the set of unicast connections is given by $\big\{(s'_t, d'_t, R_t): t \in T\big\}$.
\end{definition}

\begin{figure}
\centering
\includegraphics[width=.5\textwidth]{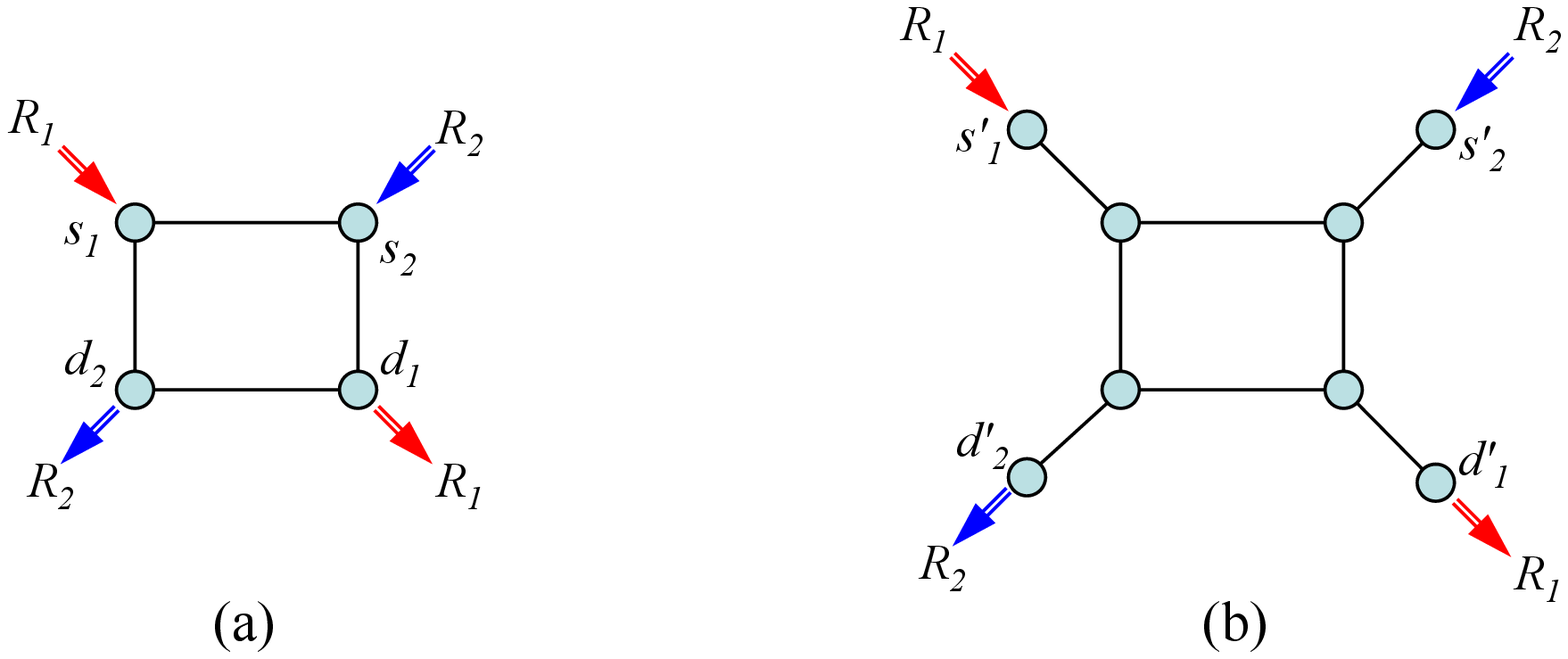}\\
      \caption{(a) A graph $G = (\mathcal N, \mathcal E)$ with two source-destination pairs (b) Expanded graph $G' = (\mathcal N', \mathcal E')$ with artificial nodes added for each source and destination.}\label{fig:artificial}
\end{figure}

Figure \ref{fig:artificial} illustrates a four node graph with two unicast connections (a) and its corresponding expanded graph (b). Every node of the expanded graph is either a source, a destination or a relay node. Moreover, all source and destination nodes of the expanded graph have degree one. This allows us to specify the conservation of flow constraints as follows:
for every $t \in T$ and every ordered pair $(i,j)$ with $\{i,j\} \in \mathcal E'$, write
\begin{eqnarray}\label{conservation_const}
   \sum_{w: (i,j,w) \in \Gamma_{G'}} x^t(i,j,w) - \sum_{v: (v,i,j) \in \Gamma_{G'}} x^t(v,i,j) = \sigma^t(i,j), &&
\end{eqnarray}
where
\begin{equation}\label{Gamma}
     \Gamma_{G'} = \Big\{(i,j,k): \{i,j\} \in \mathcal E', \{j,k\} \in \mathcal E', i \neq k \Big\}
\end{equation}
is the set of all triples of connected nodes on the expanded graph $G'$, and
 $$ \sigma^t(i,j) = \left\{
                      \begin{array}{ll}
                        R_t, & \hbox{if $i=s'_t$;} \\
                        -R_t, & \hbox{if $j = d'_t$;} \\
                        0, & \hbox{otherwise.}
                      \end{array}
                     \right.  $$

The above constraints confirm that for any unicast session $t$, and all $\{i, j\} \in \mathcal E$, in each unit time, the difference between the number of packets that travel through node $i$ to get to node $j$ and the number of packets that travel from node $i$ through node $j$ to get to some other node equals either the number of packets that originate at node $i$ if $i$ is the source of unicast $t$ or the number of packets demanded by node $j$ if $j$ is the destination of unicast $t$ or zero otherwise (see Figure \ref{fig:conservation}). For simplicity of notation, denote by $\Lambda^{(t)}$ the set of all non-negative vectors $\bs x^t$ satisfying the above conservation of flow constraints.

In order to correctly map the transmission cost on the expanded graph to that of the original graph note that each physical transmission of a particular node of the graph $G$ maps to one transmission of the same node on the expanded graph $G'$. However, transmissions of the packets in the expanded graph that result in delivery of those packets to a destination do not correspond to any physical transmission in the original graph. By subtracting the cost of the extra transmissions to the destinations of the expanded graph, one can write the total transmission cost as
    $$\sum_{i \in \mathcal N} c_i z_i - \sum_{t \in T}c_{d_t}R_t,$$
where $z_i$ is the number of physical transmissions of node $i$ on the expanded graph (cf. (\ref{tx_count})). Since, an additive constant does not change the optimal solution, we neglect the cost of retransmission at the destinations.

\begin{figure}
\centering
  \includegraphics[width=.3\textwidth]{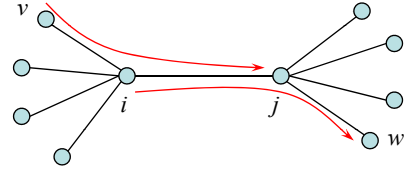}\\
  \caption{Conservation of flow constraints. The arrows demonstrate the direction of the flow of packets.}\label{fig:conservation}
\end{figure}

We can now express the problem of minimizing the total transmission cost for a multiple unicast problem with reverse carpooling as the following linear program:

\begin{eqnarray}    \label{LP}
   \textrm{minimize} \quad \sum_{i \in \mathcal N} \sum_{v<w: (v,i,w) \in \Gamma_{G'}} && \!\!\!\!\!\!\!\!\!\!\!\! c_i y^i_{vw} \\
    \textrm{subject to} \qquad \qquad \qquad \qquad && \nonumber \\
    \sum_{t \in T}x^t(v,i,w)  \le y^i_{vw}, &&\  \foral (v,i,w) \in \Gamma_{G'}, \nonumber \\
    y^i_{vw} = y^i_{wv}, && \ \foral (v,i,w) \in \Gamma_{G'}, \nonumber \\
    \bs x^t \in \Lambda^{(t)}, && \ \foral t \in T, \nonumber
\end{eqnarray}
where $\Gamma_{G'}$ is defined in (\ref{Gamma}), and $\Lambda^{(t)}$ is the set of non-negative vectors $\bs x^t$ satisfying the conservation of flow constraints (\ref{conservation_const}).
The preceding formulation of the problem is a linear optimization problem of polynomial size which can be solved efficiently in a centralized manner. However, we are interested in a distributed algorithm to solve this problem. In the following section, we use a dual decomposition together with subgradient method to solve this problem in a decentralized manner.

\section{Dual Decomposition}\label{decomposition_sec}
We consider the following Lagrangian function that penalizes the first set of constraints of the problem in (\ref{LP}). For each triple $(v,i,w) \in \Gamma_{G'}$, the Lagrange multiplier corresponding to the constraint for $(v,i,w)$ is denoted by $p(v,i,w)$.
\begin{eqnarray*}\label{lagrange}
    L(\bs x,\bs y,\bs p) &=&  \sum_{i \in \mathcal N} \sum_{v<w: (v,i,w) \in \Gamma_{G'}} \bigg[ c_i y^i_{vw} \\
    &+& p(v,i,w) \Big( \sum_{t \in T}x^t(v,i,w)-y^i_{vw}\Big)  \\
        &+& p(w,i,v) \Big( \sum_{t \in T}x^t(w,i,v)-y^i_{wv}\Big)  \bigg].
\end{eqnarray*}

The dual function is thus given by

\begin{eqnarray}\label{qp}
&& q(\bs p) =   \textrm{min} \quad  \sum_{t \in T}\sum_{(v,i,w) \in \Gamma_{G'}} p(v,i,w) x^t(v,i,w)    \\
&& + \sum_{i \in \mathcal N}\sum_{v<w, (v,i,w) \in \Gamma_{G'}} y^i_{vw}\big(c_i - p(v,i,w) - p(w,i,v)\big)  \nonumber \\
&&  \qquad  \textrm{subject to}   \nonumber \\
&& \qquad \qquad \qquad   y^i_{vw} = y^i_{wv}, \qquad \ \foral (v,i,w) \in \Gamma_{G'},  \nonumber \\
&& \qquad \qquad  \qquad   \bs  x^t \in \Lambda^{(t)}, \qquad \quad \foral t \in T. \nonumber
\end{eqnarray}

Since $y^i_{vw}$ can be positive or negative, we can drive the cost to $-\infty$ when $c_i - p(v,i,w) - p(w,i,v)$ is non-zero for any triple $(v,i,w)$. Hence, the dual function can be simplified as
\begin{equation*}\label{dual}
    q(\bs p) = \left\{
             \begin{array}{ll}
               q^*(\bs p), & \hbox{if $(c_i - p(v,i,w) - p(w,i,v)) = 0,$} \\
                 & \hbox{ \qquad \qquad \qquad$\foral (v,i,w) \in  \Gamma_{G'}$;} \\
               -\infty, & \hbox{otherwise,}
             \end{array}
           \right.
\end{equation*}
where $q^*(\bs p)$ is the optimal value of the \emph{primal sub-problem} given by
\begin{eqnarray}
 q^*(\bs p) =   \textrm{min} \quad \sum_{t \in T} \sum_{(v,i,w) \in \Gamma_{G'}} p(v,i,w) x^t(v,i,w) && \label{primal_problem} \\
   \textrm{subject to} \qquad \qquad \qquad \qquad \qquad \qquad&&  \nonumber \\
 \qquad \qquad     \bs x^t \in \Lambda^{(t)}. \qquad \foral t \in T, &&\nonumber
\end{eqnarray}

The \emph{dual problem} solves the maximization
\begin{eqnarray}
 \textrm{maximize} \quad q^*(\bs p)  \hspace{15em} &&  \label{dual_problem} \\
\textrm{subject to}   \hspace{18em}   &&  \nonumber \\
   p(v,i,w) + p(w,i,v) = c_i, \quad \foral (v,i,w) \in \Gamma_{G'},    &&\nonumber \\
  p(v,i,w) \geq 0, \quad \foral (v,i,w) \in \Gamma_{G'}.    && \nonumber
\end{eqnarray}

Note that for every feasible vector $\bs p$, $q^*(\bs p)$ provides a lower bound on the optimal cost of the problem in (\ref{LP}). Since there is no duality gap for feasible linear optimization problems (see \cite{nlp_book}), optimizing the dual problem achieves the optimal value of the original problem (\ref{LP}).

We next show how to solve the dual problem given an optimal solution to the primal sub-problem and how to decompose and solve the primal sub-problem efficiently.

\subsection{Dual Problem}
We employ the projected subgradient optimization method (cf. Section 6.3.1 of \cite{nlp_book}) to solve the dual problem given by (\ref{dual_problem}). We start from a dual feasible solution
$\bs p[0]$. Given the $n^{th}$ iteration $\bs p[n]$ for any $n\geq 0$, we compute a subgradient of $q(\bs p)$ at $\bs p[n]$, which is given by
 $$ \bs g[n] = \sum_{t \in T} \hat{\bs x}^t[n] - \hat{\bs y}[n],$$
where  $(\hat {\bs x}[n], \hat {\bs y}[n])$ is an optimal solution of the problem in (\ref{qp}) for the dual vector $\bs p = \bs p[n]$. Note that for a dual feasible solution, the values of the  $y^{i}_{vw}$ do not affect the cost or feasibility of problem (\ref{qp}). Hence, they can be selected arbitrarily, and we choose them to be zero. Therefore, $\hat {\bs x}[n]$ can be computed by solving the primal sub-problem (\ref{primal_problem}). The $(n+1)^{st}$ iterate of the subgradient method is given by
\begin{equation}\label{subgradient}
     \bs p[n+1] = \Big(\bs p[n] +  \alpha^n \sum_{t \in T} \hat{\bs x}^t[n] \Big)^+,
\end{equation}
where $\alpha^n$ is the stepsize, and $(\cdot)^+$ denotes Euclidean projection on the set of feasible solutions of the dual problem, i.e., $\bs z^+$ is the optimal solution of the following problem
\begin{eqnarray}\label{proj_problem}
 \textrm{min} \quad \sum_{v<w, (v,i,w) \in \Gamma_{G'}} (z(v,i,w) - p(v,i,w))^2 \hspace{1em} && \\+ (z(w,i,v) - p(w,i,v))^2 \hspace{1em} &&  \nonumber \\
\textrm{subject to}   \hspace{16em}   &&  \nonumber \\
   p(v,i,w) + p(w,i,v) = c_i, \quad \foral (v,i,w) \in \Gamma_{G'},    &&\nonumber \\
  p(v,i,w) \geq 0, \quad \foral (v,i,w) \in \Gamma_{G'}.    && \nonumber
\end{eqnarray}

      We can decompose the above projection problem into smaller sub-problems of projection on a two-dimensional simplex, which can be expressed in closed-form (see Figure \ref{projection_fig} for an illustration of the process). After a few straightforward steps, we can verify that for all triples $(v,i,w) \in \Gamma_{G'}$
\begin{eqnarray}\label{subgrad_projected}
    p(v,i,w)[n+1] = \Big[p(v,i,w)[n] \hspace{8em} && \nonumber \\
    + \frac{\alpha^n}{2} \sum_{t \in T} \Big(\hat x^t(v,i,w)[n] - \hat x^t(w,i,v)[n]\Big)\Big]^+_{[0,c_i]}, &&
\end{eqnarray}
where $[z]^+_{[a,b]}$ denotes projection of $z$ into interval $[a,b]$, i.e.,
$$[z]^+_{[a,b]} = \left\{
                    \begin{array}{ll}
                      a, & \hbox{if $z <a$} \\
                      z, & \hbox{if $a \le z \le b$} \\
                      b, & \hbox{if $z>b$.}
                    \end{array}
                  \right.
$$

\begin{figure}
\centering
  \includegraphics[width=.35\textwidth]{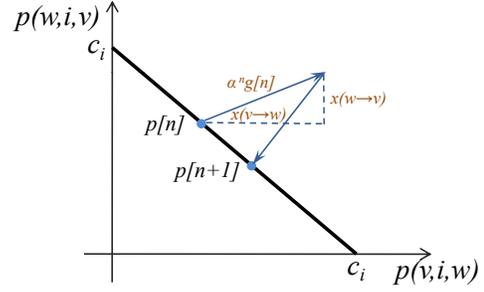}\\
  \caption{An iteration of the projected subgradient method. $x(v\rightarrow w) =  \alpha^n \sum_{t \in T} \hat{x}^t(v,i,w)[n]$ is proportional to the total flow relayed from $v$ to $w$ through $i$. Observe that when $x(v\rightarrow w) > x(w\rightarrow v)$ the price of forwarding packets from $w$ to $v$, $p(w,i,v)$, decreases to attract more flow in this direction.}\label{projection_fig}
\end{figure}

Under proper stepsize selection rules, such as the diminishing stepsize selection rule $\alpha^n = \frac1n$, the iterates given in (\ref{subgrad_projected}) converge to the dual optimal solution $\bs p^*$ as $n \rightarrow \infty$ \cite{nlp_book}. Even though the optimal solution of the primal sub-problem (\ref{primal_problem}) at $\bs p^*$ gives the optimal solution of the main problem (\ref{LP}), there is no guarantee that the corresponding primal solutions $\hat{\bs x}[n]$ approach the optimal solution of the main problem (\ref{LP}) as $\bs p[n]$ approach $\bs p^*$. However, there are several methods suggested by Sherali and Choi \cite{primal_recovery} to recover close to optimal primal solutions from the sequence of primal iterates $\{\hat{\bs x}[n]\}$. For instance, consider the sequence $\{\tilde{\bs x}[n]\}$, where
\begin{equation}\label{recovery}
\tilde{\bs x}[n] = \frac1n \sum_{k=1}^n \hat{\bs x}[k].
\end{equation}
For proper stepsize selection rules, any limit point of the sequence $\{\tilde{\bs x}[n]\}$ gives an optimal solution to the main problem (\ref{LP}) by \cite{primal_recovery}.

\subsection{Primal Sub-Problem}
We proceed to the primal sub-problem described in (\ref{primal_problem}). Since the cost function is separable and there is no constraint binding the flow variables $x^t(v,i,w)$ for different $t \in T$, we can decompose this problem into separate sub-problems. Write $q^*(\bs p) = \sum_{t \in T} q^t(\bs p)$, where for each $t \in T$

\begin{eqnarray}
 q^t(\bs p) =   \textrm{min} \quad \sum_{(v,i,w) \in \Gamma_{G'}} p(v,i,w) x^t(v,i,w) \hspace{4em}&& \label{t_subproblem} \\
   \textrm{subject to} \hspace{18em} &&  \nonumber \\
      \sum_{w: (i,j,w) \in \Gamma_{G'}} x^t(i,j,w) - \sum_{v: (v,i,j) \in \Gamma_{G'}} x^t(v,i,j) = \sigma^t(i,j), && \nonumber \\
      \foral (i,j) \textrm{ s.t. } \{i,j\} \in \mathcal E',\quad && \nonumber \\
      x^t(v,i,w) \geq 0, \quad \foral (v,i,w) \in \Gamma_{G'}.\qquad \quad \ && \nonumber
\end{eqnarray}

Each such problem is a minimum cost flow problem with no capacity constraint.  We show that this problem can be expressed as a shortest path problem  on an \emph{edge-graph} defined below.

\begin{definition}\label{hypergraph_def}
The \emph{edge-graph} $H = (\mathcal V, \mathcal A)$ for the undirected graph $G = (\mathcal N, \mathcal E)$ is defined as a weighted, directed graph where
\begin{eqnarray}
\mathcal V &=& \bigg\{ (i,j): i,j \in \mathcal N', \textrm{ and } \{i,j\} \in \mathcal E' \bigg\},
\nonumber \\
\mathcal A &=& \bigg\{ \Big((i,j),(j,k)\Big): (i,j,k) \in  \Gamma_{G'} \bigg\}. \nonumber
\end{eqnarray}
In other words, there is a directed link from node $(i,j) \in \mathcal V$  to $(k,l) \in \mathcal V$ if and only if $j=k$, and $i \neq l$.
\end{definition}

An example of an expanded graph with one source-destination pair and the corresponding edge-graph is illustrated in Figure \ref{fig:hypergraph}. Next, we establish the equivalence relation.

\begin{proposition}
For each $t \in T$, the problem in (\ref{t_subproblem}) is equivalent to the shortest path problem between nodes $(s'_t, s_t)$ and $(d'_t, d_t)$ of the edge-graph $H$ defined in Definition \ref{hypergraph_def}, when the cost of each link $ \Big((i,j),(j,k)\Big) \in \mathcal A$ is given by $p(i,j,k)$.
\end{proposition}
\begin{proof}
Consider the problem of finding the shortest path between nodes $p$ and $q$ of a directed graph $(\mathcal V, \mathcal A)$, when the cost of each link $(i,j) \in \mathcal A$ is denoted by $w_{ij}$. A linear programming formulation of this problem is
\begin{eqnarray}
 \textrm{minimize} \quad \sum_{(i,j) \in \mathcal A} w_{ij} x_{ij} \hspace{12em}&& \label{SP1} \\
   \textrm{subject to} \hspace{18em} &&  \nonumber \\
      \sum_{j: (i,j) \in \mathcal A} x_{ij} - \sum_{j: (j,i) \in \mathcal A} x_{ji} =  \sigma_i \quad \foral i \in \mathcal V, && \nonumber \\
      x_{ij} \geq 0, \quad \foral (i,j) \in \mathcal A, && \nonumber
\end{eqnarray}
where
$$\sigma_i = \left\{
  \begin{array}{ll}
    1, & \hbox{if $i = p$ ;} \\
    -1, & \hbox{if $i = q$;} \\
    0, & \hbox{otherwise.}
  \end{array}
\right.
$$

In order to see the equivalence of this problem to the problem in (\ref{t_subproblem}), it is sufficient to use the construction of the edge-graph $H$ given by Definition \ref{hypergraph_def} with $p = (s'_t, s_t)$ and $q = (d'_t, d_t)$. Note that there is a one to one correspondence
between triples $(v,i,w) \in \Gamma_{G'}$, and directed links $\big((v,i),(i,w)\big) \in \mathcal
A$, and the weight of each such triple coincides with the cost of the related link. We also need to scale the variables by $R_t$ to obtain the exact same formulation as in (\ref{t_subproblem}), but this does not affect the equivalence.
\end{proof}

       \begin{figure}
       \centering
         \includegraphics[width=.45\textwidth]{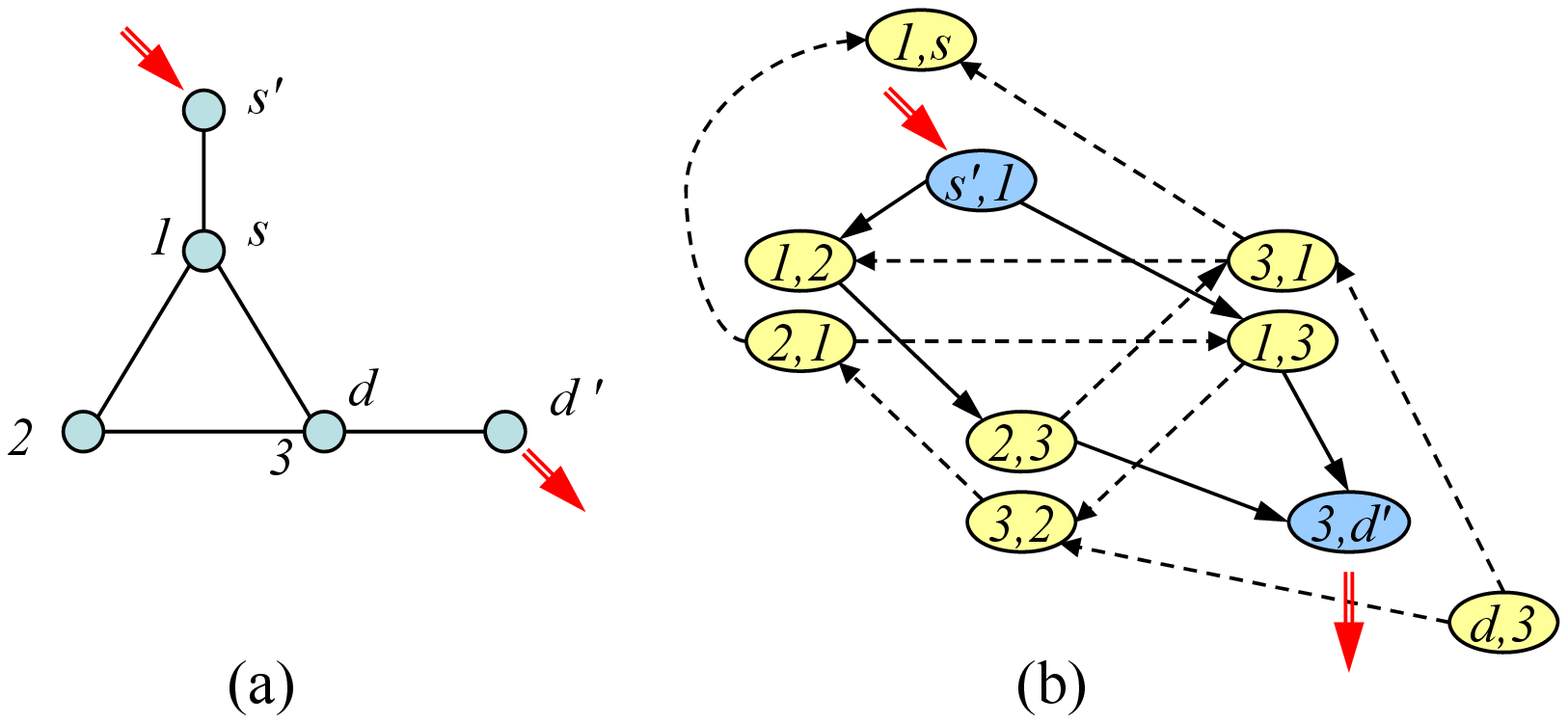}\\
         \caption{(a) Expanded graph with source node $s'$, and destination node $d'$. (b)
         Corresponding edge-graph generated as described in (\ref{hypergraph_def}) with  $(s',1)$ as the source and $(3,d')$ as the destination. The solid links can participate in a path
         from the source to the destination.   }\label{fig:hypergraph}
       \end{figure}

There are various decentralized and asynchronous implementations of the shortest path problem
available in the literature such as the well-known Bellman-Ford and Dijkstra's algorithms, (see \cite{dist_book}).
However, distributed computation of the shortest path on the edge-graph does not immediately yield
a decentralized algorithm on the original network of wireless nodes. We discuss distributed implementation of the presented algorithms on the original graph in the following section.

\section{Decentralized Implementation}\label{distributed_sec}
We provided a linear programming formulation of a minimum cost coding-aware routing problem for multiple unicast sessions. We presented an iterative algorithm based on the projected subgradient method to solve this problem. Each iteration of the algorithm is computed as in (\ref{subgrad_projected}) given an optimal solution of the primal sub-problem in (\ref{primal_problem}). We now discuss how to implement this iterative algorithm when nodes can only exchange information with their neighbors.

The dual (price) variables $p(v,i,w)$ as well as all primal (flow) variables of the form $x^t(v,i,w)$ are stored at node $i$. Given an optimal solution to the primal sub-problem, the subgradient iteration given in (\ref{subgrad_projected}) can be readily computed at node $i$ without any further coordination with other nodes.

Now consider the primal sub-problem (\ref{primal_problem}). As we discussed earlier, the primal sub-problem is equivalent to a shortest path problem on the edge-graph described in Definition \ref{hypergraph_def} for each source-destination pair. The shortest path problem on the edge-graph can be solved by the Bellman-Ford algorithm in a decentralized and asynchronous manner (cf. \cite{dist_book}). In this algorithm, every node $(i,w) \in \mathcal V$ updates its label based on the information obtained from its predecessor $(v,i) \in \mathcal V$, as well as the cost of the link between them given by $p(v,i,w)$.

For every node $(i,j)$ of the edge-graph $H$, assign its processor to node $i$ of the expanded graph $G'$, i.e., the distance information kept at node $(i,j)$ of the edge-graph is stored at node $i$. Hence, information exchange between nodes $(i,w)$ and $(v,i)$ in the edge-graph is equivalent to information exchange between nodes $i$ and $v$ in the expanded graph. This is possible because nodes $i$ and $v$ of the expanded graph are connected by definition of $\mathcal V$ in Definition \ref{hypergraph_def}. Therefore, the primal sub-problem (\ref{primal_problem}) can be solved on the expanded graph by exchanging information only with immediate neighbors. This immediately yields a distributed algorithm on the original graph $G$ by performing the computations assigned to each source (destination) node of graph $G'$ at the corresponding source (destination) node of the original graph $G$.

       \begin{figure}
       \centering
         \includegraphics[width=.35\textwidth]{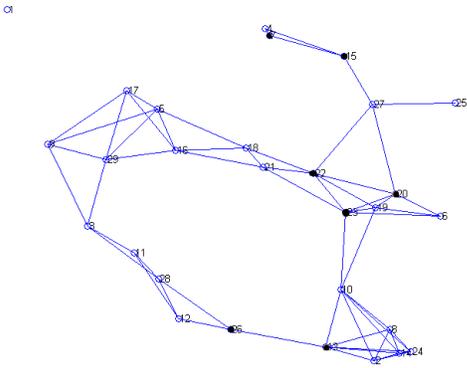}\\
         \caption{Network connectivity graph.}\label{connectivity_fig}
       \end{figure}

We can interpret the operation of the algorithm as follows. At each round, each node $i$ declares a price $p(v,i,w)$ for forwarding unit flow from $v$ to $w$, and a different price $p(w,i,v)$ for the opposite direction. Given these prices, the nodes obtain the shortest path from each source to the corresponding destination. This gives the updated flow variables, i.e., how much flow $i$ needs to relay from $v$ to $w$ and how much is relayed from $w$ to $v$. Finally, node $i$ updates the prices so that the price of forwarding packets in the direction with less flow decreases (see Figure \ref{projection_fig}). This provides an incentive for other flows to exploit more coding opportunities.

\section{Simulation Results}\label{simulation_sec}
We demonstrate the convergence of the presented iterative algorithm for solving the minimum cost coding-aware routing problem (\ref{LP}). The connectivity graph $G$ is generated by randomly placing nodes in a square according to a Poisson point process of unit rate. Two nodes are assumed  to be connected if and only if their distance is less than one. Figure \ref{connectivity_fig} illustrates a sample graph with the following source-destination pairs: $(s_1, d_1) = (20, 13); (s_2, d_2) = (26, 7); (s_3, d_3) = (15, 23); (s_4, d_4) = (7, 22)$. We assume that each pair has unit rate, and the cost of each transmission is also one.

We start the subgradient iterations in (\ref{subgrad_projected}) with the initial dual variables given by $p(v,i,w) = \frac{c_i}{2}$. We take a diminishing stepsize selection rule $\alpha^n = \frac{1}{n}$ and recover the primal feasible solutions $\tilde{\bs x}[n]$ as in (\ref{recovery}). Figure \ref{plot_fig} shows the total system cost evaluated at $\tilde{\bs x}[n]$ for each iteration of the subgradient method. The optimal value of the primal sub-problem (\ref{primal_problem}) at any dual feasible solution provides a lower bound on the optimal cost which is plotted in Figure \ref{plot_fig}.
The total transmission cost for a plain routing scheme is also plotted as a reference. Note that after a few iterations of the subgradient method the transmission cost of the system with reverse carpooling significantly improves compared to that of plain routing schemes.

     \begin{figure}
       \centering
         \includegraphics[width=.45\textwidth]{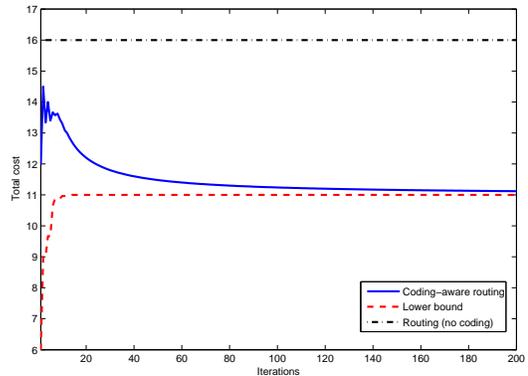}\\
         \caption{Convergence of the iterations of the subgradient method.}\label{plot_fig}
       \end{figure}

\section{Conclusions}\label{conclusion_sec}
We studied the minimum cost multiple unicast problem over a wireless network when coding opportunities of the form of reverse carpooling exist. We formulated the problem as a linear optimization problem of polynomial size by considering flow variables across triples of connected nodes. By employing dual decomposition and the projected subgradient method, we provided an iterative algorithm to obtain the optimal flow control strategies. We also showed that the presented algorithm can be implemented in a distributed manner so that each node only requires to exchange information with its neighbors.

The problem formulation in this work as well as the ones in \cite{cope_optimization} and \cite{srinivas} are based on maximizing coding opportunities to decrease the total transmission cost. However, this can cause a severe congestion in a certain area of the network, while there are other unused paths that can provide similar coding opportunities. In such scenarios, it is more desirable to distribute the flow evenly across multiple paths to avoid congestion. We shall study this problem in  the future works.


\begin{thebibliography}{1}

\bibitem{algebraic}
R. Koetter and M. M\'edard, \emph{An Algebraic Approach to Network Coding}, IEEE/ACM Trans. on
Networking, vol.11, no.5, pp. 782-795, Oct. 2003.

\bibitem{insufficiency}
R. Dougherty, C. Freiling, K. Zeger,  \emph{Insufficiency of linear coding in network information
flow},  IEEE Trans. on Information Theory, vol.51, no.8, pp. 2745- 2759, Aug. 2005.

\bibitem{tiling}
M. Effros, T. Ho, and S. Kim, \emph{A tiling approach to network code design for wireless
networks}, Information Theory Workshop, 2006.


\bibitem{cope}
S. Katti, H. Rahul, D. Katabi, W. H. M. M\'edard and J. Crowcroft,
\emph{XORs in the Air: Practical Wireless Network Coding}, ACM SIGCOMM, 2006.

\bibitem{marden}
J. Marden and M. Effros, \emph{The Price of Selfishness in Network Coding},  Workshop on Network
Coding, Theory, and Applications, 2009.

\bibitem{cope_optimization}
S Sengupta, S Rayanchu and S Banerjee, \emph{An analysis of wireless network coding for unicast
sessions: The case for coding-aware routing}, IEEE INFOCOM, 2007.

\bibitem{srinivas}
V. Reddy, S. Shakkottai, A. Sprintson and N. Gautam, \emph{Multipath Wireless Network Coding: A Population Game Perspective}, to appear in IEEE INFOCOM 2010, preliminary version at {\tt arXiv:0908.3317 [cs.GT]}.

\bibitem{dist_book}
 Dimitri P. Bertsekas and  John N. Tsitsiklis, \emph{Parallel and distributed computation: numerical methods}, Athena Scientific, 1997.

\bibitem{nlp_book}
D.P. Bertsekas., \emph{Nonlinear Programming}. Athena Scientific, 1999.


\bibitem{primal_recovery}
H. D. Sherali and G. Choi, \emph{Recovery of primal solutions when using
subgradient optimization methods to solve Langrangian duals of linear
programs}, Oper. Res. Lett., vol. 19, pp. 105–113, 1996.


\end{thebibliography}
\end{document}